\documentclass[a4paper, 10 pt, conference]{IEEEtran}  
\usepackage{amsthm}

\usepackage{amssymb}
\usepackage{amsmath}
\usepackage{color}

\usepackage{graphics} 
\usepackage{graphicx}
\usepackage{hyperref}
\graphicspath{ {./figures/} }

\title{Sparse Non-Negative Recovery from Biased Subgaussian
  Measurements using NNLS}

\author{
  \IEEEauthorblockN{
    Yonatan Shadmi, Peter Jung, Giuseppe Caire}
  \IEEEauthorblockA{
    {\small\em yonatansha@gmail.com, \{peter.jung,caire\}@tu-berlin.de}
  }
}

\newtheorem{Theorem}{Theorem}
\newtheorem{definition}{Definition}

\newtheorem{lemma}{Lemma}
\newtheorem{corollary}{Corollary}
\newcommand{\vc}[1]{\mathbf{#1}}

\begin{document}

\maketitle

\begin{abstract}

  We investigate non-negative least squares (NNLS) for the recovery of
  sparse non-negative vectors from noisy linear and biased
  measurements.  We build upon recent results from \cite{c1} showing
  that for matrices whose row-span intersects the positive orthant, the
  nullspace property (NSP) implies compressed sensing recovery
  guarantees for NNLS. Such results are as good as for
  $\ell_1$-regularized estimators but do not require tuning parameters
  that depend on the noise level.
  A  bias in the sensing matrix improves this auto-regularization feature
  of NNLS and the NSP then determines the sparse recovery performance
  only. We show that NSP holds with high probability for
  biased
  subgaussian matrices and its quality is independent of the bias.
\end{abstract}


\section{INTRODUCTION}

Compressed sensing (CS) algorithms based on $\ell_1$-regularization,
like LASSO or basis pursuit denoising (BPDN) etc., are among the most
well-known sparse recovery algorithms today, and as convex programs,
the preferred tools in many applications with well-investigated
recovery guarantees.  The idea of CS itself is based on the fact that
the intrinsic dimension of many signals or large data sets is
typically far less than their ambient dimensions, for example the
sparse representation of images, videos, audio data, network status
information like activity and novel coding techniques for wireless
communication.

However, in such practical real world applications the signals and
also the sensing matrices itself are also subject to further
constraints.  Sparsity or, more general, compressibility can be
regarded here as first order structure and the signals of interest
exhibit additional structure like block-sparsity, tree-sparsity and,
most importantly here, known sign patterns yielding to a non-negativity
constraint.  In particular such non-negative and sparse structures
also arise naturally in certain empirical inference problems, like
activity detection \cite{Caire:activitydetection,Wang:Caire},
network tomography \cite{Castro:2004}, statistical tracking (see e.g.\
\cite{Boyd2003}), compressed imaging of intensity patterns
\cite{Donoho92} and visible light communication based positioning
\cite{Gligoric:vlc:cs}.
Interestingly  non-negativity itself already
provides certain uniqueness guarantees and therefore the underlying
mathematical problem has received considerable attention in its own
right
\cite{Fuchs2005,Zhang:2010,Meinhausen,Slawsky,Foucart:NNLS:2014}.

Donoho et al. investigated in \cite{Donoho} the noiseless case in the
language of convex polytopes. They show that the sparsest non-negative
solutions can be found by convex optimization if the sparsity of the
sparsest solution is smaller than a fraction of the number of
equations.  Bruckstein, Elad and Zibulevsky \cite{Bruckstein}
investigated the uniqueness of non-negative and sparse solutions in
the noiseless case when also the entries of the sensing matrix are
non-negative, or more general, if the matrix has a row-span
intersecting the positive orthant (referred to as the
$\mathcal{M}^+$-criterion defined below).  They found that matrices belonging to the $\mathcal{M}^+$ class provide uniqueness
of non-negative sparse solutions and reconstruction therefore reduces
to just finding a feasible solution. In the noisy case, a meaningful
approach consists of replacing the search for an exact non-negative
solution of the set of linear equations by a good approximation
minimizing the residual in a certain norm (which usually depends
on further assumptions like its distribution etc.).
In the case of $\ell_2$--norms, it may be therefore
sufficient to solve the {\em Non-Negative Least squares} (NNLS).  This
relation is indeed important since NNLS requires no regularization
parameter (no additional tuning with respect to noise levels) which is often difficult to determine
for regularized problems like LASSO or basis pursuit denoising.

Slawsky and Hein discussed in \cite{Slawsky} the noisy case where the
random noise is Gaussian or subgaussian. They show (under a condition
similar to the $\mathcal{M}^+$-criterion, i.e., the self-regularizing
property) that NNLS intrinsically promotes sparsity. They further find
for subgaussian noise distributions reconstruction guarantees in the
form of upper bounds on the norm of the error vector, including the
$\ell_\infty$-case. This bound is important for hard thresholding for
sparse recovery.  Meinhausen proved in \cite{Meinhausen} similar
results for the $\ell_1$-norm of the error vector under different
assumptions on the measurements matrix. It seems also that the idea of
non-negativity as a particular conic constraint extends to other
cones. For example, Wang et al. \cite{PSDmatrices} established such
results also for the cone of positive semi-definite matrices.  Flinth
and Keiper investigated in \cite{Flinth} reconstruction of sparse
binary signals through box-constrained basis pursuit using biased
measurement matrices, and state conditions, under which the solution
can be found through box-constrained least squares instead.

In \cite{c1}, Kueng and Jung established reconstruction guarantees for
(adverserial) noise in terms of the Null Space Property (NSP) and the
$\cal{M}^+$-criterion. The conditions in \cite{c1} (NSP and the
${\cal M}^+$ criterion), play a similar role to the conditions in
\cite{Slawsky} (self regularizing property and the restricted
eigenvalues condition).  See here also \cite{Kabanava2015} for similar
steps in the low-rank matrix recovery case.  We base most of our
analysis on \cite{c1}.  Note that NSP is a sufficient and
necessary condition for the success of $\ell_1$-recovery programs like
BPDN given a correct bound on the noise power.  In contrast NNLS
always succeeds without having a-priori knowledge about the noise
power and instead the error scales in terms of the instantaneous noise
power. To have such a feature for BPDN one usually needs to
investigate the quotient property (see for example
\cite[Ch.11]{Foucart}).

{\bf Our contribution:} First, we review and extend the theory of
non-negative sparse recovery using NNLS. We consider then biased
subgaussian random $m\times n$ measurement matrices and illustrate
that a bias $\mu\geq 0$ improves the self-regularizing property of
NNLS.  However, proving NSP without taking care of the bias, for
example in using the small ball method \cite{Mendelson,c2}, yields
$\ell_2$-recovery guarantees for $s$--sparse vectors for
$m\gtrsim s\cdot(\sqrt{\log(n/s)}+\mu)^2$ observations. To overcome
this suboptimal scaling we combine a debiasing step with the small ball
method showing that NNLS has indeed the following 
(non-trivial) $\ell_q$-recovery guarantees (informal version): 

\begin{Theorem}
  Let $\vc{A}\in\mathbb{R}^{m\times n}$ be a
  random matrix whose biased rows  $\vc{a}_i=\vc{a}_{0,i}+\mu\vc{1}^T$  are independent and $\vc{a}_{0,i}$ are isotropic $1$-subgaussian  random vectors.
If for $q\geq 2$:
\[
m\gtrsim s^{2-2/q}\log(n/s)
\]
the following holds with overwelming probability:
For all $\vc{x}_0$ and $\vc{n}$, 
solution $\hat{\vc{x}}$ of the NNLS \eqref{eq:nnls} for $\vc{y}=\vc{A}\vc{x}_0+\vc{n}$ satisfies
for $p\in\left[1,q\right]$ the following error bound:
\[
\Vert\hat{\vc{x}}-\vc{x}_0\Vert_{p}\leq\frac{C}{s^{1-1/p}}\sigma_{s}\left(\vc{x}\right)_{1}+\frac{D}{s^{\frac{1}{q}-\frac{1}{p}}}\left(\frac{1}{\sqrt{m}\mu}+\tau\right)\Vert\vc{z}\Vert_2,
\]
with constants $C, D$ and $\tau$ defined below, and $\sigma_s\left(\vc{x}\right)_1$ is given in \eqref{eq:sigma}.  
\end{Theorem}
The complete statement is given in Theorem \ref{thm:main} in Section \ref{sec:main} below.

\section{SYSTEM MODEL AND OBJECTIVES}
We consider the problem of recovering a non-negative and sparse vector 
$\vc{x}_0\in\mathbb{R}_+^n$ from noisy linear observations of the form:
\[
\vc{y}=\vc{A}\vc{x}_0+\vc{n}\in\mathbb{R}^m
\]
where $\vc{A}\in\mathbb{R}^{m\times n}$ is the measurement
(sensing) matrix and $\vc{n}\in\mathbb{R}^m$ denotes additive
noise.  For example, $\vc{x}_0$ could be $s$-sparse meaning that its
non-zero entries are supported on a small subset
$S\subset\left[1\dots n\right]$ of cardinality $\left|S\right|\leq
s$.
In this work we assume that the sensing matrix $\vc{A}$ is a known
but random matrix, and we will investigate here distributions of
$\vc{A}$ which allow a robust and stable reconstruction of
$\vc{x}_0$ with overwhelming probability.  By ``robust and stable''
we mean that the error will scale appropriately with respect to the noise
power and the algorithm essentially recovers also approximately sparse vectors. Such results are known, for example,
for the basis pursuit denoising (BPDN) \cite{Foucart} which is the
following convex program:
\begin{equation}
\hat{\vc{x}}_{\text{BPDN}}=\underset{\vc{x}\in\mathbb{R}^n}{\arg \min} \Vert\vc{x}\Vert_1\hspace{1em}\textrm{s.t.}\hspace{1em}\Vert\vc{y}-\vc{A}\vc{x}\Vert_2\leq\eta, 
\label{eq:bpdn}
\end{equation}
Briefly, if $\eta$ is chosen such that
$\|\vc{y}-\vc{A}\vc{x}_0\|_2\leq\eta$ and if the sensing matrix
$\vc{A}$ has the robust null space property, there are upper bounds on reconstruction error, like
$\|\vc{x}_0-\hat{\vc{x}}_{\text{BPDN}}\|_2$, as a function of the
noise variance and the error of the best $s$-term approximation of
$\vc{x}_0$, see \cite{Foucart}.  However, in a practical setting one
often has to estimate the noise level $\eta$ first and also the
optimal tuning of compressed sensing algorithms itself is a difficult
task.  Even more, there are applications where this assumption is
critical since the noise level may depend on the unknown vector to
recover. Prototypical examples are sparse recovery under Poisson noise
model \cite{Raginsky:2010} or covariance matching problems like in
\cite{Caire:activitydetection}.

\subsection{Null Space Properties}
A well-known tool to characterize the performance of
$\ell_1$-regularized programs like the BPDN \eqref{eq:bpdn} above is
the following formulation of the robust null space property ($\ell_q$-NSP)
with respect to the $\ell_2$--norm \cite{Foucart}:
\begin{definition}\label{def-q_NSP}
  For some $q\geq 1$, a matrix $\vc{A}$ is said to have the
  $\ell_q$-robust null space property with respect to
  the $\ell_2$--norm ($\ell_q$-NSP)  with parameters $\rho\in\left(0,1\right)$
  and $\tau>0$ if $\forall\vc{v}\in\mathbb{R}^{n}$ and
  $\forall S\subset\left\{ 1,...,n\right\}$ with $\left|S\right|\leq s$:
  \begin{equation}
    \Vert\vc{v}_{S}\Vert_{q}\leq\frac{\rho}{s^{1-1/q}}\Vert\vc{v}_{S^{C}}\Vert_{1}+\tau\Vert\vc{A}\vc{v}\Vert_{2}
    \label{eq:def:NSP}
  \end{equation}
where $\vc{v}_S$ is the
vector build from $\vc{v}$ but all indices not in $S$ are set to zero, and $S^{C}$ is the complement of $S$.
\end{definition}

Due to the inequality $\Vert\vc{v}_S\Vert_p\leq s^{1/p-1/q}\Vert\vc{v}_S\Vert_q$ for $1\leq p\leq q$, the $\ell_q$-NSP implies the $\ell_p$-NSP for $p\in[1,q]$ w.r.t. the norm $s^{1/p-1/q}\Vert\cdot\Vert_2$ in the form
\begin{equation}\label{eq:p<q-NSP}
\Vert\vc{v}_{S}\Vert_{p}\leq\frac{\rho}{s^{1-1/p}}\Vert\vc{v}_{S^{C}}\Vert_{1}+s^{1/p-1/q}\tau\Vert\vc{A}\vc{v}\Vert_{2}.
\end{equation}
Let us recall the definition of the error of the best $s$--term
approximation of $\vc{x}$ in the $\ell_p$--norm which is:
\begin{equation}\label{eq:sigma}
  \sigma_{s}\left(\vc{x}\right)_{p}=\inf_{\vc{z}:\Vert\boldsymbol{z}\Vert_{0}\leq
    s}\Vert\vc{x}-\vc{z}\Vert_{p}
\end{equation}
Now, \cite[Theorem 4.25]{Foucart} states the following:
\begin{lemma}\label{thm 4.25 foucart}
  Given $1\leq p\leq q$, suppose that
  $\vc{A}\in\mathbb{R}^{m\times n}$ has the $\ell_q$-NSP
  with $0\leq\rho< 1$ and $\tau> 0$. Then, for any
  $\vc{x},\vc{z}\in\mathbb{R}^n$,
  \begin{equation*}
    \begin{split}
    \Vert\vc{z}-\vc{x}\Vert_{p}\leq\frac{C}{s^{1-1/p}}\left(\Vert\vc{z}\Vert_{1}-\Vert\vc{x}\Vert_{1}+2\sigma_{s}\left(\vc{x}\right)_{1}\right)
    \\ +Ds^{1/p-1/q}\Vert\vc{A}\left(\vc{z}-\vc{x}\right)\Vert_{2}
    \end{split}
  \end{equation*}
with $C=\frac{\left(1+\rho\right)^{2}}{1-\rho}$ , $D=\frac{\left(1+\rho\right)\tau}{1-\rho}$
\end{lemma}

Note that the requirement \eqref{eq:def:NSP} in Definition
\ref{def-q_NSP} holds for any subset $S$ of cardinality at most $s$ if
it holds for the subset $S_{\text{max}}$ containing the strongest $s$
components. Let us define $\vc{v}_s:=\vc{v}_{S_{\text{max}}}$ and
$\vc{v}_c=\vc{v}_{S_{\text{max}}^C}$.  Property \eqref{eq:def:NSP} is
also invariant with respect to re-scaling, i.e., wlog we may assume
here $\Vert\vc{v}\Vert_{q}=1$.  In addition, any vector $\vc{v}$ which
satisfies
$\Vert\vc{v}_{s}\Vert_{q}\leq\frac{\rho}{s^{1-1/q}}\Vert\vc{v}_{c}\Vert_{1}$
fulfills this condition independently of  $\vc{A}$, so we
restrict our attention to the set:
\[
T_{\rho,s}^q=\left\{ \vc{v}\in\mathbb{R}^{n}\,:\,\Vert\vc{v}_{s}\Vert_{q}>\frac{\rho}{s^{1-1/q}}\Vert\vc{v}_{c}\Vert_{1}\,,\,\Vert\vc{v}\Vert_{q}=1\right\}. 
\]
Notice, that a matrix $\vc{A}$ has the $\ell_q$--NSP of order $s$ with parameters
$\rho$ and $\tau>0$ if the following bound holds:
\begin{equation}\label{NSPcond}
\inf_{\vc{v}\in T_{\rho,s}^q}\Vert\vc{A\vc{v}}\Vert_2\geq\frac{1}{\tau}.
\end{equation}
Then we have for $\vc{v}\in T_{\rho,s}^q$:
\[
\Vert\vc{v}_{s}\Vert_{q}\leq\Vert\vc{v}\Vert_{q}=1\leq\tau\Vert\vc{A}\vc{v}\Vert_{2}\leq\frac{\rho}{s^{1-1/q}}\Vert\vc{v}_{c}\Vert_{1}+\tau\Vert\vc{A}\vc{v}\Vert_{2}.
\]
\subsection{Non-negative Sparse Recovery via NNLS}
Theorem 4 in \cite{c1} states that non-negative and $s$--sparse signals can
be robustly and stably recovered (precise statement
is Theorem \ref{thm:NNLS bound}, below) with the non-negative least squares
(NNLS):
\begin{equation}
\hat{\vc{x}}_{\text{NNLS}}=\underset{\vc{x}\in\mathbb{R}_+^n}{\arg\min}\Vert\vc{y}-\vc{A}\vc{x}\Vert_2,
\label{eq:nnls}
\end{equation}
provided that the sensing matrix $\vc{A}$ has both the $\ell_2$-NSP of
order $s$ and it satisfies the following ${\cal M}^+$-criterion.
\begin{definition}
  A matrix $\vc{A}\in\mathbb{R}^{m\times n}$ satisfies the
  ${\cal M}^+$ criterion if
  \begin{equation}
    \vc{A}\in{\cal M}^+=\left\{\vc{M}\in\mathbb{R}^{m\times n}\hspace{0.5em}\bigg|\hspace{0.5em}\exists\vc{t}\in\mathbb{R}^m\hspace{0.5em}:\hspace{0.5em}\vc{M}^T\vc{t}>\vc{0}\right\}.
    \label{eq:mplus}
  \end{equation}
\end{definition}
We like to mention here that this can be formulated for complex
matrices in the usual way.  The NNLS in \eqref{eq:nnls} has the
appealing advantage, that it does not require a-priori knowledge of
some $\eta$ such that $\Vert\vc{y}-\vc{A}\vc{x}_0\Vert_2\leq\eta$.  We repeat Theorem
4 from \cite{c1} because of its importance to this work, but first we
need to define the following condition number of a matrix as:
\begin{multline}\label{eq:con num kappa}
\kappa\left(\vc{A}\right)=\min\bigg\{  \Vert\vc{W}\Vert\Vert\vc{W}^{-1}\Vert: \\  \,\,\exists\vc{t}\in\mathbb{R}^{m}\text{with}\,\,\vc{W}=\text{diag}\left(\vc{A}^{T}\vc{t}\right)>0\bigg\} 
\end{multline}
with $\Vert\vc{W}\Vert$ being the spectral norm of the diagonal matrix $\vc{W}$.

\begin{Theorem}[Theorem 4 in \cite{c1}]
  \label{thm:NNLS bound}
  Let $\vc{A}$ be a matrix satisfying both $\ell_{q}$-NSP of order $s$
  with constants $0<\rho<1$ and $\tau>0$ and the ${\cal M}^{+}$
  criterion with $\kappa$ achieved for
  $\vc{t}$. Assume in addition that $\kappa\rho<1$, then
  for $1\leq p\leq q$:
  \[
    \Vert\hat{\vc{x}}-\vc{x}_0\Vert_{p}\leq\frac{C}{s^{1-1/p}}\sigma_{s}\left(\vc{x}\right)_{1}+
    \frac{D}{s^{\frac{1}{q}-\frac{1}{p}}}\left(\Vert\vc{t}\Vert_{2}+\tau\right)\Vert\vc{n}\Vert_2,
    \label{eq:thm:NNLS bound}
  \]
  where $C=2\frac{\kappa\left(1+\kappa\rho\right)^{2}}{1-\kappa\rho}$,
  $D=2\frac{3+\kappa\rho}{1-\kappa\rho}\max\left\{ \kappa,\Vert\vc{W}^{-1}\Vert\right\} $
  and $\hat{\vc{x}}$ is the solution of NNLS in \eqref{eq:nnls} for $\vc{y}=\vc{A}\vc{x}_0+\vc{n}$.
\end{Theorem}
Note that this theorem has been presented in \cite[Theorem 4]{c1} only
for $p=q=2$. However, its extension to $1\leq p\leq q$ is immediate by
using the fact \eqref{eq:p<q-NSP} that $\ell_p$-robust nullspace
property with respect to a norm $\|\cdot\|$ ($\ell_2$--norm in our
case) implies $\ell_q$-robust nullspace property with respect to the
norm $s^{1/p-1/q}\|\cdot\|$ (see the proof of \cite[Theorem 4]{c1} and
combine this with Lemma \ref{thm 4.25 foucart} or \cite[Theorem
4.25]{Foucart}, respectively).  Furthermore, a closer inspection of
the proof of \cite[Theorem 4]{c1} also shows that one could easily 
replace $\Vert\vc{t}\Vert_{2}$ in Theorem \ref{thm:NNLS bound} with
$\Vert\vc{t}\Vert_{2}/s^{1-1/q}$ which may have impact when
$\|\vc{t}\|_{2}$ has some undesired scaling.

\subsection{Nullspace Properties through the Small Ball Method}
A well-known tool to prove that a random matrix $\vc{A}$ (with
independent rows) has the nullspace property, i.e. \eqref{NSPcond}
holds then with high probability, is Mendelson's small ball method
\cite{Mendelson} (see here also \cite{c2}). This method is essentially
the following theorem:
\begin{Theorem}[\cite{Mendelson,c2}] \label{thm:mendelson}
Fix a set $E\subset\mathbb{R}^n$. Let the rows of a matrix $\vc{A}\in\mathbb{R}^{m\times n}$, $\vc{a}_{1},...,\vc{a}_{m}$, be independent copies
of a random vector $\vc{a}\in\mathbb{R}^n$. Define $\vc{h}=\frac{1}{\sqrt{m}}\sum_{k=1}^{m}\epsilon_{k}\vc{a}_{k}$, where
$\left\{ \epsilon_{k}\right\}_{k=1}^m $ is a Rademacher sequence\footnote{$X$ is a Rademacher variable if $X=1$ or $-1$ with equal probability.}. Then, for $t>0$ and $\xi>0$, the bound 
\[
\inf_{\vc{v}\in E}\Vert\vc{A}\vc{v}\Vert_{2}\geq\xi\sqrt{m}Q_{2\xi}\left(E,\vc{a}\right)-\xi t-2W_{m}\left(E,\vc{a}\right)
\]
with:
\[
Q_{\xi}\left(E,\vc{a}\right)=\inf_{\vc{u}\in E}\mathbb{P}\left(\left|\left\langle \vc{a},\vc{u}\right\rangle \right|\geq\xi\right)
\]
and
\[
W_{m}\left(E,\vc{a}\right)=\mathbb{E}\left[\sup_{\vc{u}\in E}\left\langle \vc{h},\vc{u}\right\rangle \right],
\]
holds with probability at least $1-\text{e}^{-2t^2}$.
\end{Theorem}
This theorem was used by Mendelson in \cite{Mendelson} in the context
of learning theory, to obtain sharp bounds on the performance of
empirical risk minimization. It was later adopted by the compressed
sensing community. Tropp used it in \cite{c2} to bound the minimum
conic singular value of matrices in certain recovery problems.  
It was used successfully in \cite{SG_extension} and \cite{c1} to
establish the NSP for certain classes of sensing matrices. In
\cite{Jung:frontiers}
a version for the complex setting has been established as well.

\subsubsection{Impact of the Bias}
In certain cases, this method -- directly applied without further adaptation --
may provide sub-optimal results.  For illustration, we bring here a
concrete example, sufficiently {\em biased} matrices, and look at the
number of measurements sufficient for stable recovery. We shall demonstrate for the
case $q=2$ that for a random matrix $\vc{A}\in\mathbb{R}^{m\times n}$,
whose entries are i.i.d. and distributed as
${\cal N}\left(\mu,1\right)$, the bound on the number of measurements obtained
from Theorem \ref{thm:mendelson} scales with $\mu$ in an
undesired manner.

To use Theorem \ref{thm:mendelson}, we take therefore the rows of $\vc{A}$ as
independent copies of $\vc{a}=\vc{g}+\mu\vc{1}$ where
$\vc{g}\sim \mathcal{N}\left(\vc{0},\vc{I}\right)$ is a standard
iid. Gaussian vector. We will only sketch the steps since this can be
found in several works \cite{SG_extension,c1,Jung:frontiers}.

First, one needs to bound
$Q_{2\xi}\left(T_{\rho,s}^2,\vc{a}\right)$. Using Paley-Zygmund
inequality, we show in Appendix \ref{Q and W bounds} that, for any
unit vector $\vc{z}\in\mathbb{R}^n$ and $\theta \in[0,1/2]$, we have:
\[
\mathbb{P}\left(\left|\vc{a}^{T}\vc{z}\right|\geq\theta\right)\geq\frac{\left(1-\theta^{2}\right)^{2}}{3}.
\]
Now, we will bound $W_{m}\left(T_{\rho,s}^2,\vc{a}\right)$. For this, first notice that $T_{\rho,s}^q$ contains all normalized $s$-sparse vectors, i.e
\[
T_{\rho,s}^q\supset\Sigma_s^q:=\left\{\vc{x}\in\mathbb{R}^n:\Vert\vc{x}\Vert_0\leq s,\Vert\vc{x}\Vert_q=1\right\}. 
\]
A converse result is also known (Lemma 3.2 in \cite{SG_extension}):
\[
T_{\rho,s}^q\subset\left(2+\rho^{-1}\right)\textnormal{conv}\left(\Sigma_s^q\right)\subseteq\frac{3}{\rho}\textnormal{conv}\left(\Sigma_s^q\right),
\]
from which we can have the bound
\begin{equation} \label{W_T bound}
W_{m}\left(T_{\rho,s}^q,\vc{a}\right)\leq \frac{3}{\rho}W_{m}\left(\Sigma_s^q,\vc{a}\right)\leq\frac{3s^{1/2-1/q}}{\rho}W_{m}\left(\Sigma_s^2,\vc{a}\right),
\end{equation}
where the last bound is due to the relation $\Vert \vc{x}\Vert_2\leq s^{1/2-1/q}\Vert \vc{x}\Vert_q$ for $s$-sparse vectors (Hoelder inequality).
In Appendix \ref{Q and W bounds} we then show that 
\[
W_{m}\left(\Sigma_{s}^{2},\vc{a}\right) \leq\sqrt{2s\log\left(\frac{n}{s}\right)+2s}+\mu\sqrt{2s}.
\]
One could think that this bound is too pessimistic, but Figure
\ref{fig:simulating W} shows indeed the linear dependency of
$W_{m}\left(T_{\rho,s}^{2},\vc{a}\right)$ in $\mu$, which shifts
$W_{m}\left(T_{\rho,s}^{2},\vc{a}\right)$ away from zero.
Nevertheless, summarizing the results, for $\xi=\theta/2$:
\begin{align*}
\inf_{\vc{v}\in T_{\rho,s}^2} & \Vert\vc{A}\vc{v}\Vert_{2}  \geq\xi\sqrt{m}Q_{2\xi}\left(T_{\rho,s}^2,\vc{a}\right)-\xi t-2W_{m}\left(T_{\rho,s}^2,\vc{a}\right) \\
& \geq\frac{\theta\left(1-\theta^2\right)^2\sqrt[]{m}}{6}-\frac{\theta}{2} t \\ & \hspace{3em}-\frac{6}{\rho}\left(\sqrt{2s\log\left(\frac{n}{s}\right)+2s}+\mu\sqrt{2s}\right).
\end{align*}
If we choose, for example, $\xi^2=1/8$ and $t=\sqrt{m}/24$, we have that the right hand side is positive if:
\[
m\geq\left(\frac{6\cdot 24\sqrt{2}}{\rho}\right)^22s\left(\sqrt{\log\left(\frac{en}{s}\right)}+\mu\right)^2
\]


Thus, from this result one might think that $\mu$ directly affects the
number of measurements. A main purpose of our work is to show that this is
not the case.  
\begin{figure}
\includegraphics[width=\columnwidth]{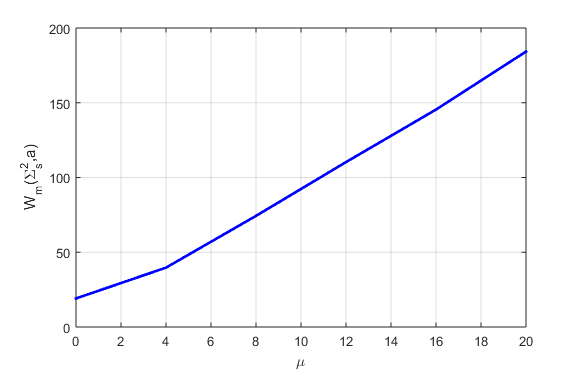}\centering
\caption{Computing numerically via simulation the width
  $W_{m}\left(\Sigma_s^{2},\vc{a}\right)$ as a
  function of $\mu$ for $s=128$ and $n=1000$.}
\label{fig:simulating W}
\end{figure}

In the next section we will therefore present a debiased
version of Theorem \ref{thm:mendelson} which indeed brings back
the known results on
the number of measurements for centered matrices, independent of
$\mu$ but at the cost of doubling $m$.

\section{Main results}
We have shown above that usual small ball method for establishing the
nullspace property suffers from a bias $\mu$ in the measurement matrices
whereby a bias will be essential for self-regularizing property (the
$\mathcal{M}^+$-criterion in \eqref{eq:mplus}) of NNLS.
Our main result, however, shows that this is not the case and NSP is
essentially independent of bias. To establish this result we present
next a debiased version of the small ball result.

\subsection{Debiased Mendelson's small balls method}
The following results parallels Theorem \ref{thm:mendelson} but
handles better a bias $\vc{e}\in\mathbb{R}^n$ (constant offset) in the
random measurement vectors (rows of the measurement matrix).
\begin{Theorem}\label{modified mendelson}
Fix a set $E\subset\mathbb{R}^n$. Let the rows of a matrix $\vc{A}\in\mathbb{R}^{m\times n}$, $\vc{a}_{1},...,\vc{a}_{m}$, be independent copies of a random vector $\vc{a}$. Define the matrix $\vc{B}\in\mathbb{R}^{\left\lfloor m/2\right\rfloor\times n}$ whose rows are $\vc{b}_i=\left(\vc{a}_{2i-1}-\vc{a}_{2i}\right)/\sqrt{2}$ and  $\vc{h}=\frac{1}{\sqrt{\lfloor m/2\rfloor}}\sum_{k=1}^{\lfloor m/2\rfloor}\epsilon_{k}\vc{b}_k$, where
$\left\{ \epsilon_{k}\right\} $ is a Rademacher sequence. Then, for $t>0$ and $\xi>0$, The bound 
\[
\inf_{\vc{v}\in E}\Vert\vc{A}\vc{v}\Vert_{2}\geq\xi\sqrt{\frac{m-1}{2}}Q_{2\xi}\left(E,\vc{b}\right)-\xi t-2W_{\lfloor m/2\rfloor}\left(E,\vc{b}\right)
\]
with $Q_{\xi}$ and $W_{m}$ as defined in Theorem \ref{thm:mendelson},
holds with probability at least $1-\text{e}^{-2t^{2}}$.
\end{Theorem}
A first version, proving the results in \cite{Shadmi:isit19} and only valid for
distributions which are symmetric around $\vc{e}$, has been presented
by the authors already in the first arXiv version
\cite{Shadmi:nnls:arxiv:v1} of this work.

\begin{proof}
  The following technique is motivated by a similar debiasing step in
  \cite[Sec. IV(b)]{Chen:2015}.
  \begin{align*}
    \left\Vert\vc{Av}\right\Vert_2^2&=\sum_{i=1}^m\left|\left<\vc{a}_i,\vc{v}\right>\right|^2
                                      \geq\sum_{i=1}^{2\left\lfloor m/2\right\rfloor}\left|\left<\vc{a}_i,\vc{v}\right>\right|^2\\
                                    &=\sum_{i=1}^{\left\lfloor m/2\right\rfloor}\left(\left|\left<\vc{a}_{2i-1},\vc{v}\right>\right|^2+\left|\left<\vc{a}_{2i},\vc{v}\right>\right|^2\right)\\
                                    &\overset{(a)}{\geq}\frac{1}{2}\sum_{i=1}^{\left\lfloor m/2\right\rfloor}\left(\left|\left<\vc{a}_{2i-1},\vc{v}\right>\right|+\left|\left<\vc{a}_{2i},\vc{v}\right>\right|\right)^2\\
                                    &\overset{(b)}{\geq}\frac{1}{2}\sum_{i=1}^{\left\lfloor m/2\right\rfloor}\left(\left|\left<\vc{a}_{2i-1},\vc{v}\right>-\left<\vc{a}_{2i},\vc{v}\right>\right|\right)^2\\
                                    &=\sum_{i=1}^{\left\lfloor m/2\right\rfloor}\left|\left<\frac{1}{\sqrt{2}}\left(\vc{a}_{2i-1}-\vc{a}_{2i}\right),\vc{v}\right>\right|^2\\
                                    &=\sum_{i=1}^{\left\lfloor m/2\right\rfloor}\left|\left<\vc{b}_i,\vc{v}\right>\right|^2
                                    =\left\Vert\vc{Bv}\right\Vert_2^2,
\end{align*}
where in (a) we used Cauchy-Schwarz inequality and (b) the reverse
triangle inequality.  The last missing step is just an application of
Theorem \ref{thm:mendelson} on the matrix $\vc{B}$, and the conclusion
follows.
\end{proof}
This theorem will allow to prove the NSP of biased matrices even if
$\mu\propto\sqrt{m}$, and even for any fixed bias, it improves the
previous bounds (up to constants).

\begin{corollary}\label{NSP of SG}
  Let $\vc{A}\in\mathbb{R}^{m\times n}$ be a matrix whose rows can be written in the form $\vc{a}_i=\vc{a}_{0,i}+\mu\vc{1}$ where $\{\vc{a}_{0,i},1\leq i\leq m\}$ are
  independent copies of a random vector $\vc{a}_0$ with the following properties:
  \begin{itemize}
  \item is sub-isotropic, i.e. $\mathbb{E}\left[\left\langle\vc{a}_0,\vc{v}\right\rangle^2\right]\geq\Vert\vc{v}\Vert_2^2$ for all $\vc{v}\in\mathbb{R}^n$,
  \item 1-subgaussian, i.e., $\mathbb{E}\left[\exp\left(t\left\langle\vc{a}_0,\vc{v}\right\rangle\right)\right]\leq\exp\left(t^2\right)$ for all $\vc{v}\in\mathbb{R}^n$ with $\Vert\vc{v}\Vert_2^2\leq 1$ and $t\in\mathbb{R}$,
  \end{itemize}
  then, $\vc{A}$ has the $\ell_{q}$-NSP for all $q\geq 2$ with probability
  at least $1-\textnormal{e}^{-\left(m-1\right)/64}$ if
  \begin{equation}
    m\geq \frac{2\cdot384^2}{\rho^2}s^{2-2/q}\left(2+\hspace{0.5em}\sqrt[]{\log\left(\frac{\textnormal{e}n}{s}\right)}\right)^2+1.
\label{eq:phasetrans:debiased}
\end{equation}
\end{corollary}
\begin{proof}
We will show that Equation \eqref{NSPcond} holds by using Theorem \ref{modified mendelson} and bounding $\inf_{\vc{u}\in T^q_{\rho,s}}\left\Vert\vc{Bu}\right\Vert_2$. By the definition of $\vc{B}$, the fact that the rows of $\vc{A}$ are sub-Gaussians and Theorem 7.27 in \cite{Foucart}, we know that the rows of $\vc{B}$ are also sub-Gaussian with the same sub-Gaussian norm. In addition:
\begin{align*}
    \mathbb{E}\left[\left<\vc{b_i},\vc{v}\right>^2\right]&=\mathbb{E}\left[\left<\frac{\vc{a}_{2i-1}-\vc{a}_{2i}}{\sqrt{2}},\vc{v}\right>^2\right]\\
    &=\mathbb{E}\left[\left<\frac{\vc{a}_{0,2i-1}-\vc{a}_{0,2i}}{\sqrt{2}},\vc{v}\right>^2\right]\\
    &=\frac{1}{2}\mathbb{E}\left[\left<\vc{a}_{0,2i-1},\vc{v}\right>^2\right]+\frac{1}{2}\mathbb{E}\left[\left<\vc{a}_{0,2i},\vc{v}\right>^2\right]\\
    &\quad-\mathbb{E}\left[\left<\vc{a}_{0,2i},\vc{v}\right>\left<\vc{a}_{0,2i-1},\vc{v}\right>\right]\\
    &=\frac{1}{2}\mathbb{E}\left[\left<\vc{a}_{0,2i-1},\vc{v}\right>^2\right]+\frac{1}{2}\mathbb{E}\left[\left<\vc{a}_{0,2i},\vc{v}\right>^2\right]\\
    &\geq\left\Vert\vc{v}\right\Vert_2^2,
\end{align*}
so the rows of $\vc{B}$ are sub-isotropic as well. Thus, we can use results from the proof of Corollary 5.2 from \cite{SG_extension}. There, they first use the sub-isotropic and sub-Gaussian properties to bound $Q_\xi\left(T_{\rho,s}^q,\vc{b}\right)$. Notice that for $q\geq 2$ and some $\vc{u}\in T_{\rho,s}^q$ we have $1=\Vert\vc{u}\Vert_q\leq\Vert\vc{u}\Vert_2$. Therefore:
\begin{align*}
\mathbb{P}\left(\left|\left\langle\vc{b},\vc{u}\right\rangle\right|>\xi\right) & =\mathbb{P}\left(\left|\left\langle\vc{b},\frac{\vc{u}}{\Vert\vc{u}\Vert_2}\right\rangle\right|>\frac{\xi}{\Vert\vc{u}\Vert_2}\right) \\
& \geq\mathbb{P}\left(\left|\left\langle\vc{b},\frac{\vc{u}}{\Vert\vc{u}\Vert_2}\right\rangle\right|>\xi\right).
\end{align*}
With Paley-Zygmund inequality we then have \cite[proof of Corollary 5.2]{SG_extension}:
\[
\mathbb{P}\left(\left|\left\langle\vc{b},\frac{\vc{u}}{\Vert\vc{u}\Vert_2}\right\rangle\right|>\xi\right)\geq\left(1-\xi^2\right)^2
\]
for $0\leq\xi\leq 1$.
Recall equation (\ref{W_T bound}), and combine it with the following bound:
\[
W_{\lfloor m/2\rfloor}\left(\Sigma_s^2,\vc{b}\right)\leq 4\sqrt{2}\left(2\hspace{0.5em}\sqrt[]{s}+\hspace{0.5em}\sqrt[]{s\log\left(\frac{\textnormal{e}n}{s}\right)}\right).
\]
The last bound is due to Dudley's inequality, taken from \cite[Theorem 8.23]{Foucart}, and, again, from the proof of Corollary 5.2 from \cite{SG_extension}. 
Now, by choosing for example $t=\sqrt{m-1}/\left(8\sqrt{2}\right)$ and $\xi^2=1/2$, we have with probability
  at least $1-\textnormal{e}^{-\left(m-1\right)/64}$ (Theorem \ref{modified mendelson}):
\begin{align*}
    \inf_{\vc{v}\in T^q_{\rho,s}}\left\Vert\vc{Av}\right\Vert_2&\geq\xi\sqrt{\frac{m-1}{2}}Q_{2\xi}\left(T^q_{\rho,s},\vc{b}\right)\\
    &\hspace{20pt}-\xi t-2W_{\left\lfloor m/2\right\rfloor}\left(T^q_{\rho,s},\vc{b}\right)\\
    &\geq\xi\sqrt{\frac{m-1}{2}}\left(1-\xi^2\right)^2-\xi t\\
    &\quad-\frac{3s^{1/2-1/q}}{\rho}8\sqrt{2}\left(2\hspace{0.5em}\sqrt[]{s}+\hspace{0.5em}\sqrt[]{s\log\left(\frac{\textnormal{e}n}{s}\right)}\right)\\
    &=\frac{\sqrt{m-1}}{16}-\frac{3s^{1-1/q}}{\rho}8\sqrt{2}\left(2+\hspace{0.5em}\sqrt[]{\log\left(\frac{\textnormal{e}n}{s}\right)}\right).
\end{align*}
If the bound is positive, we can set it to $1/\tau$ and we have the NSP. This bound is positive if
\begin{align*}
    m\geq \frac{2\cdot384^2}{\rho^2}s^{2-2/q}\left(2+\hspace{0.5em}\sqrt[]{\log\left(\frac{\textnormal{e}n}{s}\right)}\right)^2+1.
\end{align*}
\end{proof}

Notice that this bound is independent of the bias now.

\subsection{Establishing the $\cal{M}^+$ criterion}

\begin{Theorem}\label{M+ for SG}
Let $\vc{A}\in\mathbb{R}^{m\times n}$ be a random matrix with independent rows. Assume that the columns of $\vc{A}$ can be written as 
$\vc{a}^{col}_i=\vc{a}^{col}_{0,i}+\mu\vc{1}$ for $1\leq i\leq n$, where $\vc{a}^{col}_{0,i}$ are zero-mean
1-subgaussian random vectors (with independent components), i.e., satisfying
$\mathbb{E}\left[\exp\left(t\left\langle\vc{a}^{col}_{0,i},\vc{v}\right\rangle\right)\right]\leq\exp\left(t^2\right)$
 for all $\vc{v}\in\mathbb{R}^m$ with $\Vert\vc{v}\Vert_2^2\leq 1$ and $t\in\mathbb{R}$ for all $ 1\leq i\leq n$. Then $\vc{A}\in\cal{M}^+$ with probability at least 
$1-2n\exp\left(-\frac{\mu^2m}{16}\right)$.
\end{Theorem}

\begin{proof}
We choose the vector $\vc{t}$ to be 
$\vc{t}=\frac{1}{m\mu}\vc{1}$ and have $\vc{w}=\vc{A}^T\vc{t}$. Compute now
\begin{align*}
\left|\text{w}_i-1\right| &=
\left|(\vc{A}^T\vc{t})_i-1\right| \\ &=\left|\frac{1}{m\mu}\sum_{k=1}^mA_{k,i}-1\right| \\
&= \left|\frac{1}{m\mu}\sum_{k=1}^m\left( A_{k,i}-\mu\right)\right| \\
&= \left|\frac{1}{m\mu}\sum_{k=1}^m\left( A_{k,i}-\mathbb{E}\left[A_{k,i}\right]\right)\right|.
\end{align*}
We can now use Hoeffding's inequality (Theorem 7.27 from \cite{Foucart}) to bound this term by:
\[
\mathbb{P}\left(\left|\frac{1}{m\mu}\sum_{k=1}^m\left(A_{k,i}-\mathbb{E}\left[A_{k,i}\right]\right)\right|\geq\frac{1}{2}\right)\leq2\exp\left(-\frac{\mu^2m}{16}\right).
\]
This guarantees that $\text{w}_i$ is positive with probability better than the one written above, and by applying the union bound we have that all the components of $\vc{w}$ are positive with probability at least 
$1-2n\exp\left(-\frac{\mu^2m}{16}\right)$.
\end{proof}
Thus, for $\mu\rightarrow\infty$ the probability
$\mathbb{P}\left(\vc{A}\in\cal{M}^+\right)$ converges to 1.
This also proves that $\kappa\left(\vc{A}\right)\leq 3$ with probability as in Theorem \ref{M+ for SG}.

\subsection{Main theorem}\label{sec:main}
Now we combine here the results for NSP and $M^+$ criterion.
\begin{Theorem}\label{thm:main}
 Set $\rho\in\left(0,1/3\right)$, $q\geq 2$ and
  $p\in\left[1,q\right]$. Let $\vc{A}\in\mathbb{R}^{m\times n}$ be a
  random matrix whose rows can be written as $\vc{a}_i=\vc{a}_{0,i}+\mu\vc{1}^T$, where $\{\vc{a}_{0,i},1\leq i\leq m\}$ are independent copies of a random vector $\vc{a}_0$ which is:
\begin{itemize}
\item sub-isotropic, i.e., $\mathbb{E}\left[\left\langle\vc{a}_0,\vc{v}\right\rangle^2\right]\geq\Vert\vc{v}\Vert_2^2$ for all $\vc{v}\in\mathbb{R}^n$,
\item 1-subgaussian, i.e., $\mathbb{E}\left[\exp\left(t\left\langle\vc{a}_0,\vc{v}\right\rangle\right)\right]\leq\exp\left(t^2\right)$ for all $\vc{v}\in\mathbb{R}^n$ with $\Vert\vc{v}\Vert_2^2\leq 1$ and $t\in\mathbb{R}$,
\end{itemize} 
and whose columns can be written for all $i$ as
$\vc{a}^{col}_i=\vc{a}^{col}_{0,i}+\mu\vc{1}$ with $\vc{a}^{col}_{0,i}$ also being 1-subgaussian.
If
\[
m\geq \frac{2\cdot384^2}{\rho^2}s^{2-2/q}\left(2+\hspace{0.5em}\sqrt[]{\log\left(\frac{\textnormal{e}n}{s}\right)}\right)^2+1,
\]
the following holds with probability at least
$1-\textnormal{e}^{-\left(m-1\right)/64}-2n\exp\left(-\frac{\mu^2m}{16}\right)$:
For all $\vc{x}_0$ and $\vc{n}$, 
solution $\hat{\vc{x}}$ of the NNLS \eqref{eq:nnls} for $\vc{y}=\vc{A}\vc{x}_0+\vc{n}$ satisfies the following
error bound:
\[
\Vert\hat{\vc{x}}-\vc{x}_0\Vert_{p}\leq\frac{C}{s^{1-1/p}}\sigma_{s}\left(\vc{x}\right)_{1}+\frac{D}{s^{\frac{1}{q}-\frac{1}{p}}}\left(\frac{1}{\sqrt{m}\mu}+\tau\right)\Vert\vc{n}\Vert_2,
\]
with constants $C, D$ and $\tau$ defined as in Theorem \ref{thm:NNLS bound}.  
\end{Theorem}
\begin{proof}
This is just an application of the union bound to bound the probability of the intersection of the events that Corollary \ref{NSP of SG} and Theorem \ref{M+ for SG} hold together, since the conditions of both are assumed to be satisfied. Since both the $\cal{M}^+$ criterion and the $\ell_2$-NSP hold, NNLS reconstructs the original vector with a reconstruction guarantee according to Theorem \ref{thm:NNLS bound}. Notice that we used our choice of $\vc{t}$ from the proof of Theorem \ref{M+ for SG} to compute $\Vert\vc{t}\Vert_2$, and the condition $\kappa\rho<1$ is satisfied with high probability because $\kappa<3$ (proof of Theorem \ref{M+ for SG}).
\end{proof}

\section{NUMERICAL EXPERIMENTS}
In this section we provide numerical experiments to support the
results of the previous sections. We measured the recovery performance
in terms of the normalized square roof of the MSE, given by:
\[
\text{NMSE}=\frac{\Vert\hat{\vc{x}}-\vc{x}_0\Vert_2}{\|\vc{x}_0\|_2}
\]
for four algorithms: NNLS \eqref{eq:nnls} with biased and with centered sensing
matrix, BPDN \eqref{eq:bpdn} with biased and with centered sensing matrix.  For NNLS
we have used either the internal {\sc Matlab} routine ``lsqnonneg''
which is based on the active-set algorithm \cite{Lawson74} or a
speed-optimized version ``bwhiten'' \footnote{ B. Whiten, ``nnls -
  Alternative to lsqnonneg'',
  \url{https://de.mathworks.com/matlabcentral/fileexchange/38003-nnls-non-negative-least-squares}},
and BPDN has been solved using the {\sc Cvx}-toolbox.

For fixed $n=100$ and $s=5$, we take
$m\in\left\{20,25,...,55,60,70,80 \right\}$, and randomly generate
Gaussian matrices with i.i.d. entries $\mathcal{N}(\mu,1)$ for
$\mu\in\left\{0,20\right\}$. The non-negative signals were either binary vectors
$\vc{x}_0\in\{0,1\}^n$ or absolute value of
a standard normal random variables. We add Gaussian white noise $\vc{n}$
with
zero mean and variance $\sigma^2=-20$dB to the measurements, and reconstruct the
signal, either by NNLS \eqref{eq:nnls} or by BPDN \eqref{eq:bpdn}
using the instantaneous norm $\eta=\|\vc{n}\|_2$. Note that this
already reflects some instantaneous extra knowledge for BPDN.
\begin{figure}
\includegraphics[width=\columnwidth]{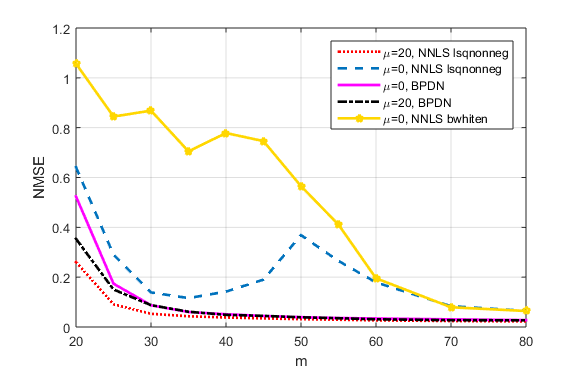}\centering
\caption{NMSE of NNLS in reconstruction of non-negative sparse binary signals, with Gaussian sensing matrices.\label{fig:NMSEbinary}}
\includegraphics[width=\columnwidth]{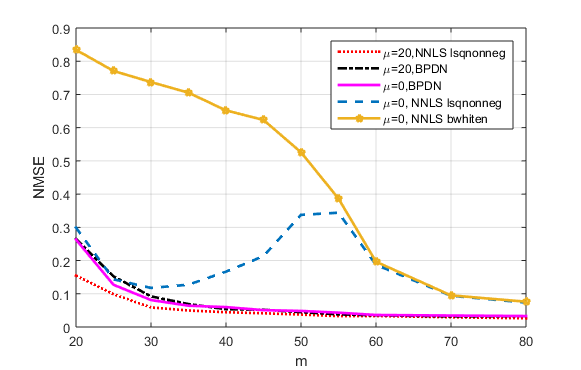}\centering
\caption{NMSE of NNLS in reconstruction of non-negative sparse one-sided-Gaussian signals, with Gaussian sensing matrices.\label{fig:NMSEabsgaussian}}
\end{figure}

The results are given in Figures \ref{fig:NMSEbinary} and \ref{fig:NMSEabsgaussian}. First, we see that the
bias is critical for reconstruction with NNLS, since non-biased
matrices are likely not to belong to $\cal{M}^+$. Therefore, NNLS with centered
matrices performed worse than the other three algorithms. It is known
that in general NNLS makes sense for this case
only when $m\geq n/2$ (see for example comments in \cite{PSDmatrices}).
For $m\leq n/2$ there is with high probability
no unique solution in the noiseless case, i.e., the NNLS performance
is determined by the algorithm implementation.
In contrast, NNLS reconstruction of a non-negative
signal from biased measurements, where the sensing matrix both
satisfies the NSP and belong to $\cal{M}^+$, achieved the best
performance among the four.

Another thing worth noticing, is  the equivalence between biased and
centered matrices when using BPDN, when $m$ is above some
threshold. This can be expected from Corollary \ref{NSP of SG}, since, by this result, the bias plays no role for the NSP.

\section{Conclusions}

We obtained recovery guarantees for NNLS in the case of biased subgaussian matrices for non-negative sparse vectors. For that purpose, we needed to show that these matrices satisfy the NSP. For this, we first used Mendelson's small ball method, and saw that the bias affects the bound in a negative way. We showed that the NSP of the biased matrix is implied by the NSP of a related centered matrix. This allows to ignore the bias and to find better bounds on the class of biased matrices, even when the bias is much bigger than the variance.



\section*{Acknowledgements}
We thank Dominik Stoeger and Saeid Haghighatshoar for fruitful discussions.
The work was funded by the Alexander-von-Humboldt
foundation and partially supported by DAAD grant 57417688.
PJ has been supported by DFG grant JU 2795/3.


\section*{APPENDIX}
\subsection{Bounds on $Q_{2\xi}\left(T_{\rho,s}^2,\vc{a}\right)$ and
  $W_m\left(\Sigma_s^2,\vc{a}\right)$ }
\label{Q and W bounds}

We can use Paley-Zygmund inequality to bound $Q_{2\xi}\left(T_{\rho,s}^2,\vc{a}\right)$
from below. Let $\vc{z}$ be a unit vector, and $\theta\in\left[0,1/2\right]$, and recall that $\vc{a}=\vc{g}+\mu\vc{1}$ where $\vc{g}\sim{\cal N}\left(\vc{0},\vc{I}\right)$. Define $S=\left(\vc{a}^{T}\vc{z}\right)^{2}$. Then for any unit vector $\vc{z}$:
\begin{align*}
S & =\left(\vc{a}^{T}\vc{z}\right)^{2}=\left(\mu\vc{1}^{T}\vc{z}+\vc{g}^{T}\vc{z}\right)^{2} \\ & =\mu^{2}\vc{z}^{T}\vc{1}\vc{1}^{T}\vc{z}+\vc{z}^{T}\vc{g}\vc{g}^{T}\vc{z}+2\mu\vc{1}^{T}\vc{z}\vc{z}^{T}\vc{g}.
\end{align*}
Now,
\[
\mathbb{E}\left[S\right]=\mu^{2}\vc{z}^{T}\vc{1}\vc{1}^{T}\vc{z}+\Vert\vc{z}\Vert_{2}^{2}\geq 1,
\]
and therefore
\[
\mathbb{P}\left(\left|\vc{a}^{T}\vc{u}\right|\geq\theta\right)\geq\mathbb{P}\left(S\geq\theta^{2}\mathbb{E}\left[S\right]\right)\geq\left(1-\theta^{2}\right)^{2}\frac{\mathbb{E}\left[S\right]^{2}}{\mathbb{E}\left[S^{2}\right]}.
\]
In addition,
\begin{align*}
S^{2}= & \mu^{4}\left(\vc{z}^{T}\vc{1}\vc{1}^{T}\vc{z}\right)^{2}+\left(\vc{z}^{T}\vc{g}\right)^{4}+ \\ & 4\mu^{2}\left(\vc{1}^{T}\vc{z}\right)^{2}\vc{z}^{T}\vc{g}\vc{g}^{T}\vc{z}+2\mu^{2}\vc{z}^{T}\vc{1}\vc{1}^{T}\vc{z}\vc{z}^{T}\vc{g}\vc{g}^{T}\vc{z}+ \\ & 4\mu^{3}\left(\vc{1}^{T}\vc{z}\right)^{3}\vc{z}^{T}\vc{g}+4\mu\vc{1}^{T}\vc{z}\left(\vc{z}^{T}\vc{g}\right)^{3},
\end{align*}

and since $\vc{z}^{T}\vc{g}\sim{\cal N}\left(0,1\right)$:
\begin{align*}
\mathbb{E}\left[S^{2}\right] & =\mu^{4}\left(\vc{z}^{T}\vc{1}\vc{1}^{T}\vc{z}\right)^{2}+3+4\mu^{2}\left(\vc{1}^{T}\vc{z}\right)^{2}\Vert\vc{z}\Vert_{2}^{2} \\
& \hspace{8em}+2\mu^{2}\vc{z}^{T}\vc{1}\vc{1}^{T}\vc{z}\Vert\vc{z}\Vert_{2}^{2} \\
& =\mu^{4}\left(\vc{z}^{T}\vc{1}\vc{1}^{T}\vc{z}\right)^{2}+3+6\mu^{2}\left(\vc{1}^{T}\vc{z}\right)^{2}\Vert\vc{z}\Vert_{2}^{2}.
\end{align*}

Finally:
\begin{align*}
\frac{\mathbb{E}^{2}\left[S\right]}{\mathbb{E}\left[S^{2}\right]} & =\frac{\mu^{4}\left(\vc{1}^{T}\vc{z}\right)^{4}+2\mu^{2}\left(\vc{1}^{T}\vc{z}\right)^{2}+1}{\mu^{4}\left(\vc{1}^{T}\vc{z}\right)^{4}+6\mu^{2}\left(\vc{1}^{T}\vc{z}\right)^{2}+3} \\
& =1-\frac{4\mu^{2}\left(\vc{1}^{T}\vc{z}\right)^{2}+2}{\mu^{4}\left(\vc{1}^{T}\vc{z}\right)^{4}+6\mu^{2}\left(\vc{1}^{T}\vc{z}\right)^{2}+3},
\end{align*}

and since 
\[
\mu^{4}\left(\vc{1}^{T}\vc{z}\right)^{4}+6\mu^{2}\left(\vc{1}^{T}\vc{z}\right)^{2}+3\geq 6\mu^{2}\left(\vc{1}^{T}\vc{z}\right)^{2}+3
\]
then
\[
\frac{4\mu^{2}\left(\vc{1}^{T}\vc{z}\right)^{2}+2}{\mu^{4}\left(\vc{1}^{T}\vc{z}\right)^{4}+6\mu^{2}\left(\vc{1}^{T}\vc{z}\right)^{2}+3}\leq\frac{2}{3}
\]
and then 
\[
\frac{\mathbb{E}^{2}\left[S\right]}{\mathbb{E}\left[S^{2}\right]}=1-\frac{4\mu^{2}\left(\vc{1}^{T}\vc{z}\right)^{2}+2}{\mu^{4}\left(\vc{1}^{T}\vc{z}\right)^{4}+6\mu^{2}\left(\vc{1}^{T}\vc{z}\right)^{2}+3}\geq\frac{1}{3}
\]
and
\[
\mathbb{P}\left(\left|\vc{a}^{T}\vc{u}\right|\geq\theta\right)\geq\left(1-\theta^{2}\right)^{2}\frac{\mathbb{E}\left[S\right]^{2}}{\mathbb{E}\left[S^{2}\right]}\geq\frac{\left(1-\theta^{2}\right)^{2}}{3}.
\]

We turn now to look at $W_m\left(\Sigma_s^2,\vc{a}\right)$. Denote:
\[
\vc{h}_{0}=\frac{1}{\sqrt{m}}\sum_{k=1}^{m}\epsilon_{k}\vc{g}_{k}
\]
with $\vc{g}_{k}\sim{\cal N}\left(\vc{0},\vc{I}\right)$,
and from symmetry of $\epsilon_{k}$ and the normal distribution, $\vc{h}_{0}\sim{\cal N}\left(\vc{0},\vc{I}\right).$
\begin{align*}
\Longrightarrow\mathbb{E} & \left[\sup_{\vc{u}\in\Sigma_{s}^{2}}\left\langle\vc{h},\vc{u}\right\rangle\right] =\mathbb{E}\left[\sup_{\vc{u}\in\Sigma_{s}^{2}}\left\langle\vc{h}_{0},\vc{u}\right\rangle+\frac{\mu}{\sqrt{m}}\sum_{k=1}^{m}\epsilon_{k}\left\langle\vc{1},\vc{u}\right\rangle\right] \\ &
\leq\mathbb{E}\left[\sup_{\vc{u}\in\Sigma_{s}^{2}}\left\langle\vc{h}_{0},\vc{u}\right\rangle\right]+\mathbb{E}\left[\sup_{\vc{u}\in\Sigma_{s}^{2}}\frac{\mu}{\sqrt{m}}\sum_{k=1}^{m}\epsilon_{k}\left\langle\vc{1},\vc{u}\right\rangle\right].
\end{align*}

The first term is proportional to the Gaussian width of $\Sigma_{s}^{2}$
so we can use known results (Proposition 3.10 in \cite{Gausswidth}). For the second term
\begin{align*}
& \mathbb{E}\left[\sup_{\vc{u}\in\Sigma_{s}^{2}}\frac{\mu}{\sqrt{m}}\sum_{k=1}^{m}\epsilon_{k}\left\langle\vc{1},\vc{u}\right\rangle\right] \\ 
 \leq & \mathbb{E}\left[\sup_{\vc{u}\in\Sigma_{s}^{2}}\frac{\mu}{\sqrt{m}}\left|\sum_{k=1}^{m}\epsilon_{k}\right|\Vert\vc{1}\Vert_{\infty}\Vert\vc{u}\Vert_{1}\right] \\
& \leq\frac{\mu}{\sqrt{m}}\sqrt{s}\mathbb{E}\left[\left|\sum_{k=1}^{m}\epsilon_{k}\right|\right]\leq\sqrt{2s}\mu,
\end{align*}

where the last inequality follows from a Khintchine argument (see
\cite{Foucart}, Corollary 8.7), and from bounds on the Gaussian width
of the set $\Sigma_s^2$ (Proposition 3.10 in \cite{Gausswidth}) we
get
\begin{align*}
W_{m}\left(\Sigma_{s}^{2},\vc{a}\right) & =\mathbb{E}\left[\sup_{\vc{u}\in\Sigma_{s}^{2}}\left\langle\vc{h},\vc{u}\right\rangle\right] \\ & \leq\left(\sqrt{2s\log\left(\frac{n}{s}\right)+2s}+\mu\sqrt{2s}\right).
\end{align*}

\bibliographystyle{IEEEtran}
\bibliography{main}

\begin{thebibliography}{10}
\providecommand{\url}[1]{#1}
\csname url@samestyle\endcsname
\providecommand{\newblock}{\relax}
\providecommand{\bibinfo}[2]{#2}
\providecommand{\BIBentrySTDinterwordspacing}{\spaceskip=0pt\relax}
\providecommand{\BIBentryALTinterwordstretchfactor}{4}
\providecommand{\BIBentryALTinterwordspacing}{\spaceskip=\fontdimen2\font plus
\BIBentryALTinterwordstretchfactor\fontdimen3\font minus
  \fontdimen4\font\relax}
\providecommand{\BIBforeignlanguage}[2]{{%
\expandafter\ifx\csname l@#1\endcsname\relax
\typeout{** WARNING: IEEEtran.bst: No hyphenation pattern has been}%
\typeout{** loaded for the language `#1'. Using the pattern for}%
\typeout{** the default language instead.}%
\else
\language=\csname l@#1\endcsname
\fi
#2}}
\providecommand{\BIBdecl}{\relax}
\BIBdecl

\bibitem{c1}
R.~Kueng and P.~Jung, ``{Robust Nonnegative Sparse Recovery and the Nullspace
  Property of 0/1 Measurements},'' \emph{IEEE Transactions on Information
  Theory}, vol.~64, no.~2, pp. 689--703, feb 2018.

\bibitem{Caire:activitydetection}
S.~Haghighatshoar, P.~Jung, and G.~Caire, ``Improved scaling law for activity
  detection in massive mimo systems,'' \emph{IEEE International Symposium on
  Information Theory (ISIT)}, 2018.

\bibitem{Wang:Caire}
C.~Wang, O.~Y. Bursalioglu, H.~Papadopoulos, and G.~Caire, ``On-the-fly
  large-scale channel-gain estimation for massive antenna-array base
  stations,'' in \emph{2018 IEEE International Conference on Communications
  (ICC)}, May 2018, pp. 1--6.

\bibitem{Castro:2004}
R.~Castro, M.~Coates, G.~Liang, R.~Nowak, and B.~Yu, ``{Network Tomography:
  Recent Developments},'' \emph{Stat. Sci.}, vol.~19, pp. 499--517, 2004.

\bibitem{Boyd2003}
J.~E. Boyd and J.~Meloche, ``{Evaluation of statistical and multiple-hypothesis
  tracking for video traffic surveillance},'' \emph{Mach. Vision Appl.},
  vol.~13, no. 5-6, pp. 344--351, 2003.

\bibitem{Donoho92}
D.~L. Donoho, I.~M. Johnstone, J.~C. Hoch, and S.~A. S, ``{Maximum Entropy and
  the Nearly Black Object},'' \emph{J. Roy. Stat. Soc. B Met.}, vol.~54, no.~1,
  1992.

\bibitem{Gligoric:vlc:cs}
K.~{Gligoric}, M.~{Ajmani}, D.~{Vukobratović}, and S.~{Sinanović}, ``Visible
  light communications-based indoor positioning via compressed sensing,''
  \emph{IEEE Communications Letters}, vol.~22, no.~7, pp. 1410--1413, July
  2018.

\bibitem{Fuchs2005}
J.-J. Fuchs, ``Sparsity and uniqueness for some specific under-determined
  linear systems,'' in \emph{IEEE International Conference on Acoustics,
  Speech, and Signal Processing, 2005}, vol.~5.\hskip 1em plus 0.5em minus
  0.4em\relax IEEE, 2005, pp. 729--732.

\bibitem{Zhang:2010}
G.~Zhang, S.~Jiao, X.~Xu, and L.~Wang, ``{Compressed sensing and reconstruction
  with Bernoulli matrices},'' in \emph{IEEE International Conference on
  Information and Automation}, 2010, pp. 455--460.

\bibitem{Meinhausen}
N.~Meinshausen, ``{Sign-constrained least squares estimation for
  high-dimensional regression},'' \emph{Electron. J. Stat.}, vol.~7, no.~1, pp.
  1607--1631, 2013.

\bibitem{Slawsky}
\BIBentryALTinterwordspacing
M.~Slawski and M.~Hein, ``Non-negative least squares for high-dimensional
  linear models: Consistency and sparse recovery without regularization,''
  \emph{Electron. J. Statist.}, vol.~7, pp. 3004--3056, 2013. [Online].
  Available: \url{https://doi.org/10.1214/13-EJS868}
\BIBentrySTDinterwordspacing

\bibitem{Foucart:NNLS:2014}
S.~Foucart and D.~Koslicki, ``{Sparse Recovery by Means of Nonnegative Least
  Squares},'' \emph{IEEE Signal Proc. Let.}, vol.~21, no.~4, 2014.

\bibitem{Donoho}
D.~L. Donoho and J.~Tanner, ``{Sparse nonnegative solution of underdetermined
  linear equations by linear programming.}'' \emph{P. Natl. Acad. Sci. USA},
  vol. 102, no.~27, pp. 9446--9451, 2005.

\bibitem{Bruckstein}
A.~Bruckstein, M.~Elad, and M.~Zibulevsky, ``Sparse non-negative solution of a
  linear system of equations is unique,'' 04 2008, pp. 762 -- 767.

\bibitem{PSDmatrices}
M.~Wang, W.~Xu, and A.~Tang, ``{A unique "nonnegative" solution to an
  underdetermined system: From vectors to matrices},'' \emph{IEEE Trans.
  Inform. Theory}, vol.~59, no.~3, pp. 1007--1016, 2011.

\bibitem{Flinth}
A.~Flinth and S.~Keiper, ``Recovery of binary sparse signals with biased
  measurement matrices,'' 01 2018.

\bibitem{Kabanava2015}
M.~Kabanava, R.~Kueng, H.~Rauhut, and U.~Terstiege, ``{Stable low-rank matrix
  recovery via null space properties},'' \emph{Information and Inference},
  vol.~5, no.~4, pp. 405--441, dec 2016.

\bibitem{Foucart}
S.~Foucart and H.~Rauhut, \emph{A {M}athematical {I}ntroduction to
  {C}ompressive {S}ensing}, ser. Applied and Numerical Harmonic Analysis.\hskip
  1em plus 0.5em minus 0.4em\relax Birkh\"auser/Springer, New York, 2013.

\bibitem{Mendelson}
\BIBentryALTinterwordspacing
S.~Mendelson, ``Learning without concentration,'' \emph{CoRR}, vol.
  abs/1401.0304, 2014. [Online]. Available:
  \url{http://arxiv.org/abs/1401.0304}
\BIBentrySTDinterwordspacing

\bibitem{c2}
\BIBentryALTinterwordspacing
J.~A. Tropp, ``Convex recovery of a structured signal from independent random
  linear measurements,'' \emph{CoRR}, vol. abs/1405.1102, 2014. [Online].
  Available: \url{http://arxiv.org/abs/1405.1102}
\BIBentrySTDinterwordspacing

\bibitem{Raginsky:2010}
M.~{Raginsky}, R.~M. {Willett}, Z.~T. {Harmany}, and R.~F. {Marcia},
  ``Compressed sensing performance bounds under poisson noise,'' \emph{IEEE
  Transactions on Signal Processing}, vol.~58, no.~8, pp. 3990--4002, Aug 2010.

\bibitem{SG_extension}
\BIBentryALTinterwordspacing
S.~Dirksen, G.~Lecu{\'{e}}, and H.~Rauhut, ``On the gap between rip-properties
  and sparse recovery conditions,'' \emph{CoRR}, vol. abs/1504.05073, 2015.
  [Online]. Available: \url{http://arxiv.org/abs/1504.05073}
\BIBentrySTDinterwordspacing

\bibitem{Jung:frontiers}
\BIBentryALTinterwordspacing
P.~Jung, R.~Kueng, and D.~Mixon, ``{Derandomizing compressed sensing with
  combinatorial design}.'' [Online]. Available:
  \url{http://arxiv.org/abs/1812.08130}
\BIBentrySTDinterwordspacing

\bibitem{Shadmi:isit19}
Y.~Shadmi, P.~Jung, and G.~Caire, ``{Sparse Non-Negative Recovery from Shifted
  Symmetric Subgaussian Measurements using NNLS},'' in \emph{IEEE Int.
  Symposium on Information Theory (ISIT)}, 2019.

\bibitem{Shadmi:nnls:arxiv:v1}
\BIBentryALTinterwordspacing
------, ``{Sparse Non-Negative Recovery from Biased Subgaussian Measurements
  using NNLS}.'' [Online]. Available: \url{http://arxiv.org/abs/1901.05727v1}
\BIBentrySTDinterwordspacing

\bibitem{Chen:2015}
Y.~{Chen}, Y.~{Chi}, and A.~J. {Goldsmith}, ``Exact and stable covariance
  estimation from quadratic sampling via convex programming,'' \emph{IEEE
  Transactions on Information Theory}, vol.~61, no.~7, pp. 4034--4059, July
  2015.

\bibitem{Lawson74}
C.~L. Lawson and R.~J. Hanson, \emph{{Solving Least Squares Problems}}.\hskip
  1em plus 0.5em minus 0.4em\relax Prentice-Hall, 1974.

\bibitem{Gausswidth}
V.~Chandrasekaran, B.~Recht, P.~A. Parrilo, and A.~S. Willsky, ``The convex
  geometry of linear inverse problems,'' \emph{arXiv e-prints}, p.
  arXiv:1012.0621, Dec. 2010.

\end{thebibliography}

\end{document}